\documentclass[a4paper, UKenglish,cleveref, autoref, thm-restate, final]{lipics-v2021}
\usepackage{amsmath}
\usepackage{amsthm}
\usepackage{mathtools}
\usepackage{algorithm,algpseudocode}
\usepackage{tikz}
\usepackage{listings}
\newcommand{\cut}{\mathrm{cut}}

\nolinenumbers
\newcommand\restr[2]{{
  \left.\kern-\nulldelimiterspace 
  #1 
  \vphantom{\big|} 
  \right|_{#2} 
  }}
\allowdisplaybreaks[4]

\usepackage{blindtext}

\newtheorem{fact}[theorem]{Fact}

\newtheorem*{thm*}{Theorem}
\newtheorem*{prop*}{Proposition}
\newtheorem*{obs*}{Observation}
\newtheorem*{lemma*}{Lemma}

\newtheorem*{rec*}{Recommendation}

\newenvironment{fminipage}%
  {\begin{Sbox}\begin{minipage}}%
  {\end{minipage}\end{Sbox}\fbox{\TheSbox}}

\DeclareMathOperator{\polylog}{polylog}

\newcommand{\newclass}[2]{\newcommand{#1}{{\text{\upshape\sffamily #2}}\xspace}}

\newclass{\NP}{NP}
\newclass{\ZPP}{ZPP}
\newclass{\coNP}{coNP}
\newclass{\BPP}{BPP}
\newclass{\Logspace}{L}
\newclass{\NL}{NL}
\newclass{\coNL}{coNL}
\newclass{\UL}{UL}
\newclass{\coUL}{coUL}
\newclass{\BPL}{BPL}
\newclass{\PL}{PL}
\newclass{\prBPL}{prBPL}
\newclass{\PSPACE}{PSPACE}
\newclass{\EXP}{EXP}
\newclass{\EXPSPACE}{EXPSPACE}
\newclass{\TIME}{TIME}
\newclass{\SPACE}{SPACE}
\newclass{\NSPACE}{NSPACE}
\newclass{\SC}{SC}
\newclass{\coNSPACE}{coNSPACE}
\newclass{\BPSPACE}{BPSPACE}
\newclass{\TFNP}{TFNP}

\newclass{\NC}{NC}
\newclass{\NCo}{NC$^1$}
\newclass{\ACz}{AC$^0$}
\newclass{\ACo}{AC$^1$}
\newclass{\SACo}{SAC$^1$}
\newclass{\TC}{TC}
\newclass{\TCz}{TC$^0$}
\newclass{\TCo}{TC$^1$}
\newclass{\NCt}{NC$^2$}
\newclass{\RNC}{RNC}
\newclass{\PSDNC}{pseudo-deterministic NC}
\newclass{\NUSPL}{non-uniform SPL}
\newclass{\RNCt}{RNC$^2$}
\newclass{\RNCtt}{RNC$^3$}
\newclass{\RNCo}{RNC$^1$}
\newclass{\QNC}{Quasi-NC}
\newclass{\VP}{VP}
\newclass{\CL}{CL}
\newclass{\CLP}{CLP}
\newclass{\CSPACE}{CSPACE}

\newclass{\GapL}{GapL}

\newclass{\MATCH}{MATCH}
\newclass{\DET}{DET}
\newclass{\LOSSY}{LOSSY}
\newclass{\LOSSYNC}{LOSSY[NC]}
\newclass{\LOSSYC}{LOSSY[$\mathcal{C}$]}

\newclass{\ZPC}{ZP-$\mathcal{C}$}
\newclass{\ZPNC}{ZPNC}
\newclass{\Comp}{Comp}
\newclass{\Decomp}{Decomp}
\newcommand{\I}{\mathbb{I}}
\newcommand{\Ind}[1]{\I\left[#1\right]}

\authorrunning{A. Agarwala and N. Varma}
\Copyright{Aryan Agarwala and Nithin Varma}
\newif\ifblind

\blindtrue

\title{Pseudodeterministic Algorithms for Minimum Cut Problems}

\author{Aryan {Agarwala}}{Max-Planck-Institut f\"{u}r Informatik, Saarland Informatics Campus, Germany} {aryan@agarwalas.in}{https://orcid.org/0000-0001-7047-2650}{}
\author{Nithin {Varma}}{University of Cologne, Germany}{nithvarma@gmail.com}{https://orcid.org/0000-0002-1211-2566}{}

\ccsdesc{Mathematics of computing~Graph algorithms}\ccsdesc{Theory of computation~Streaming, sublinear and near linear time algorithms}\ccsdesc{Theory of computation~Pseudorandomness and derandomization}
\keywords{Minimum Cut, Pseudodeterministic Algorithms}
\acknowledgements{The authors thank Danupon Nanongkai for extensive discussions.}
\ArticleNo{94}
\begin{document}

\setlength{\abovedisplayskip}{5pt}
\setlength{\belowdisplayskip}{5pt}

\maketitle

\begin{abstract}
    In this paper, we present efficient pseudodeterministic algorithms for both the global minimum cut and minimum $s$-$t$ cut problems. The running time of our algorithm for the global minimum cut problem is asymptotically better than the fastest sequential deterministic global minimum cut algorithm (Henzinger, Li, Rao, Wang; SODA 2024).
    Furthermore, we implement our algorithm in sequential, streaming, \textsf{PRAM}, and cut-query models, where no efficient deterministic global minimum cut algorithms are known.
\end{abstract}

\newcommand{\psdst}{\mathsf{PD}_{s,t}^{\mathcal{A}}}
\newcommand{\psdgb}{\mathsf{PD}_\mathrm{global}^{\mathcal{A}}}
\newcommand{\es}{E^{s}_{\mathrm{star}}}
\newcommand{\esi}[1]{E_{\mathrm{star}}^{< #1}}
\newcommand{\wst}{w_{s}}
\newcommand{\wi}[1]{w_{s}^{< #1}}
\newcommand{\uqtest}{\textsc{Uniqueness-Test}^\mathcal{A}}

\section{Introduction}

Randomization is one of the cornerstones in the design of algorithms and has been effective in the design of simple and fast algorithms for a variety of computational problems. The advantage of deterministic algorithms, on the other hand, is their \emph{replicability}, where multiple runs of the same algorithm on the same input result in the same output. The question of whether randomness results in provably faster algorithms than their deterministic counterparts is fundamental to the field of theoretical computer science. \\ 

\noindent The investigation of pseudodeterminism, initiated by Gat and Goldwasser~\cite{GatGoldwasser2011} has implications to the aforementioned line of work~\cite{GoldreichGoldwasserRon2013,LuOS21}. A pseudodeterministic algorithm for a search problem is a randomized algorithm which outputs the same answer with high probability. For the global minimum cut problem, a randomized algorithm only needs to output with high probability \textit{some} minimum cut, while a pseudodeterministic algorithm needs to output, with high probability, the \textit{same} minimum cut. 
Formally, an algorithm $\mathcal{A}$ is pseudodeterministic if for all inputs $x$, with probability $\rho \geq \frac{2}{3}$, we have $\mathcal{A}(x) = s(x)$, where $s(x)$ is a canonical solution for $x$. 
Pseudodeterministic algorithms capture the replicability property of deterministic algorithms, while being as efficient and simple as randomized algorithms, thereby giving us the best of both worlds. Pseudodeterministic algorithms have been studied for the construction of primes~\cite{OliveiraS17,ChenLuOliveiraRenSanthanam23,Oliveira24}, bipartite matching~\cite{GoldwasserGrossman17}, non-bipartite matching~\cite{AnariVazirani19}, matroid intersection~\cite{GhoshG21}, undirected connectivity~\cite{GrossmanLiu19}, and more. It is an important open direction as to which problems have efficient pseudodeterministic algorithms. 

\noindent \subsection{Pseudodeterministic Global Minimum Cut Algorithm.} In this work, we investigate the design of pseudodeterministic algorithms for the well-studied problem of global minimum cut. 
The global minimum cut problem takes as input a weighted undirected graph $G$  on $n$ vertices and $m$ edges, and produces as output a partition of the vertices of $G$ which minimises the total weight of edges crossing the partition. This is closely related to the minimum $s$-$t$ cut problem, which produces as output a partition of vertices of $G$ such that the vertices $s$ and $t$ lie on different sides of the partition. Naively, the global minimum cut problem can be solved using $n-1$ calls to a minimum $s$-$t$ cut algorithm~\cite{GomoryHu61}. Thus, it was initially believed that global minimum cut is a harder variant of the minimum $s$-$t$ cut problem. Eventually, Nagamochi and Ibaraki~\cite{NagamochiIbaraki92} and Hao and Orlin~\cite{HaoOrlin94} showed that global minimum cut can be solved with running time matching the best known minimum $s$-$t$ cut algorithms at the time. Karger then initiated a line of work on \textit{randomized} algorithms for the global minimum cut problem~\cite{Karger94, KargerStein96, Karger00}, eventually culminating in a $O(m\log^3n)$ time algorithm, which is faster than the best minimum $s$-$t$ algorithms even 25 years later. Thus, global minimum cut became one of the classical examples of how randomness can lead to significantly faster and simpler algorithms. Gawrychowski et al.~\cite{GawrychowskiMozesWeimann24} improved a log factor in Karger's algorithm, giving a $O(m \log^2 n)$ time randomized algorithm for global minimum cut. Recently, Henzinger et al.~\cite{HenzingerLiRaoWang24} obtained an $\tilde{O}(m)$\footnote{The $\tilde{O}$ hides $\polylog(n)$ factors.} time deterministic algorithm for global minimum cut. We note that their algorithm, while matching the runtime of randomized algorithms upto polylog factors, seems to have a high exponent in the log term and is far more complicated than Karger's near-linear time algorithm~\cite{Karger00}. \\

\noindent We present a simple and efficient pseudodeterministic algorithm for the global minimum cut problem that uses, as black-box, any randomized algorithm for global minimum cut.  

\begin{theorem}
\label{thm:sequential}
    Consider a randomized algorithm $\mathcal{A}$ that takes as input a simple graph $G = (V,E)$ with edge weights given by $w: E \to \mathbb{Z}^+$, runs in time $t(n, m)$ and outputs, with probability at least $\frac{2}{3}$, a minimum cut of $G$ with respect to $w$. Then there exists a pseudodeterministic algorithm $\mathcal{A}'$ that, on input $G = (V,E,w)$,  runs in time $O(t(n, m+n) \log (n) \log \log (n))$, and 
    with probability at least $\frac{2}{3}$, outputs a specific global minimum cut of $G$.   
\end{theorem}

Our result, in conjunction with the algorithm of Gawrychowski et al.~\cite{GawrychowskiMozesWeimann24}, immediately implies a pseudodeterministic global minimum cut algorithm with running time $O(m \log^3 n \log \log n)$, which is significantly faster than the best deterministic algorithm. The relative simplicity of our algorithm is an additional advantage.

\noindent
\subsection{Streaming, Cut Query, and Parallel Models.} 
Since its introduction, pseudodeterminism has been studied in many contexts, such as query complexity~\cite{GoldreichGoldwasserRon2013, ChattopadhyayDahitaMahajan23, GoldwasserImpagliazzoPitassiSanthanam2021}, streaming algorithms~\cite{GrossmanGuptaSellke23, BravermanKrauthgamerKrishnanSapir2023, GoldwasserGrossmanMohantyWoodruff2020}, parallel algorithms~\cite{GoldwasserGrossman17, AnariVazirani19, GhoshGurjarRaj23}, space-bounded computation~\cite{GrossmanLiu19}, and more\footnote{For a better overview of the field, we refer the reader to the doctoral thesis of Grossman~\cite{Grossman2023}.}.
In this paper, we provide efficient implementations of our pseudodeterministic global minimum cut algorithm in the streaming, cut query and parallel models. 
Our pseudodeterministic algorithms are significantly more efficient than the best-known deterministic algorithms in the respective models. In fact, as we discuss below, there exist large gaps between randomized and deterministic algorithms in all these models, unlike in the sequential model. 

\begin{enumerate}
    \item \textbf{Streaming}: In this model, the input graph consists of a stream of insertions or deletions of edges and their associated weights. Specifically, each stream element is either of the form \emph{insert $(e, w_e)$}, where $e \in E$ and $w_e \in \mathbb{Z}^+$, or of the form \emph{delete $e$}. Note that we consider streaming of simple graphs, and therefore, an edge once inserted cannot be inserted again until it is deleted. The typical goal is to compute the global minimum cut using as little space and making as few passes over this stream as possible.

    A randomized exact streaming algorithm for global minimum cut using $\tilde{O}(n)$ space and $\log(n)$ passes was shown by Mukhopadhyay and Nanongkai~\cite{MukhopadhyayNanongkai20}. As far as we are aware, there is no deterministic algorithm for this problem using $O(n^{1.99})$ space and making $\tilde{O}(1)$ many passes. 

    \begin{theorem}
    \label{thm:streaming}
        There exists a pseudodeterministic streaming algorithm for global minimum cut that makes $O(\log^2 n)$ passes over the input and uses $\tilde{O}(n)$ space. 
    \end{theorem}
    
    
    \item \textbf{Cut Query}: Here, access to the graph is permitted only via an oracle which takes as input a partition of the vertices into two parts and produces as output the total weight of edges which cross this partition. The goal is to compute the global minimum cut making as few queries to this oracle as possible.

    A randomized exact algorithm for global minimum cut using $\tilde{O}(n)$ many cut queries was shown by Mukhopadhyay and Nanongkai~\cite{MukhopadhyayNanongkai20}. As far as we are aware, the best deterministic algorithm is a corollary of the algorithm by Grebinski and Kucherov~\cite{GrebinskiKucherov00} and reconstructs the entire graph using $O(n^2/\log n)$ many queries.

    \begin{theorem}
    \label{thm:cut-query}
        There is a pseudodeterministic global minimum cut algorithm that makes $\tilde{O}(n)$ cut queries.
    \end{theorem}

    \item \textbf{Parallel}: We consider the Common-Read-Exclusive-Write ($\mathsf{CREW}$) Parallel Random Access Machine ($\mathsf{PRAM}$) model, where there are $p$ processors connected to a shared random access memory. An arbitrary number of processors can read from a memory cell, whereas only one processor can write on a memory cell in a single time step. 

    In this model, a randomized exact algorithm for global minimum cut with $\tilde{O}(m)$ work and $\tilde{O}(1)$ depth was given by Anderson and Blelloch~\cite{AndersonBlelloch23}. As far as we are aware, the only deterministic $\NC$ algorithm for this problem is by Karger and Motwani~\cite{KargerMotwani97}. While they do not explicitly calculate the work of their algorithm, it is certainly a huge polynomial, and seems to be worse than even, say, $O(n^{10})$.

    \begin{theorem}
    \label{thm:parallel}
        There is a pseudodeterministic global minimum cut algorithm in the PRAM model that does $\tilde{O}(m)$ work and is of depth $\tilde{O}(1)$. 
    \end{theorem}
\end{enumerate}

\subsection{Our Techniques}

A key tool in our algorithm is built upon the \textit{isolation lemma}. Introduced by Mulmuley, Vazirani, and Vazirani~\cite{MulmuleyVaziraniVazirani87}, the isolation lemma states that for any family of objects built on a ground set $E$, such as perfect matchings of a bipartite graph, if one assigns random polynomially bounded integer weights to each element of $E$, the minimum weight object, i.e., perfect matching, will be unique with high probability. Derandomizing the isolation lemma has since been a major open problem, with a long line of work on it (See e.g.\ \cite{ChariRohatiSrinivasan93, DattaKulkarniRoy10, FennerGurjarThierauf16, SvenssonTarnawski17, AgarwalaMertz25})  \\

\noindent
Goldwasser and Grossman~\cite{GoldwasserGrossman17}, in their pseudodeterministic $\NC$ algorithm for bipartite perfect matching, gave a pseudodeterministic algorithm for assigning edge weights to a bipartite graph such that the minimum weight perfect matching is unique with high probability. Then, using a randomized $\NC$ algorithm for minimum weight perfect matching on the input graph with their constructed edge weights, they obtain a pseudodeterministic $\NC$ algorithm for bipartite perfect matching. Anari and Vazirani~\cite{AnariVazirani19} do the same for non-bipartite matching, and Ghosh et al.~\cite{GhoshGurjarRaj23} for linear matroid intersection.\\

\noindent
In this work, we follow a similar route as the algorithm of~\cite{GoldwasserGrossman17}. Our pseudodeterministic algorithm takes as input a weighted graph, and applies perturbations to the weights in order to make the global minimum cut of the graph unique. We mention that the previous algorithms that follow this approach in~\cite{GoldwasserGrossman17, AnariVazirani19} are slower than their randomized counterparts by a polynomial factor. Thus, ours is the first instance of this approach that preserves computational efficiency.\\

\noindent The implementations of our algorithm in the four models that we discussed (sequential, streaming, cut query, parallel) lose only a logarithmic factor in efficiency in all of those models. Thus, combining our algorithm with the randomized algorithms in ~\cite{GawrychowskiMozesWeimann24, MukhopadhyayNanongkai20, AndersonBlelloch23}, we obtain near optimal pseudodeterministic algorithms for global minimum cut across many models of computation. Moreover, while the isolation lemma has previously found applications in parallel~\cite{MulmuleyVaziraniVazirani87, DattaKulkarniRoy10, FennerGurjarThierauf16, SvenssonTarnawski17}, space-bounded~\cite{ReinhardtAllender00, vanMelkebeekPrakriya19}, catalytic~\cite{AgarwalaMertz25, LiPyneTell24} and sequential~\cite{MulmuleyVaziraniVazirani87} models, this is the first application of isolation that we are aware of in streaming and cut query models.\\

\noindent
We note that our exact approach can also give similar blackbox statements for the more general problem of submodular function minimisation. For succinctness, we do not state this explicitly in the paper.

\subsection{Organization}
The paper is organized as follows. \Cref{sec:prelims} contains basic notation and definitions that we use throughout. \Cref{sec:s-t-mincut} presents a pseudodeterministic algorithm for the minimum $s$-$t$ cut problem and introduces some of the key ideas that we use in this paper. \Cref{sec:global-min-cut} contains our pseudodeterministic global minimum cut algorithm and its analysis. Finally, \Cref{sec:implementation} contains the details of the implementation of our algorithm in the streaming, cut query and parallel models. 


\section{Preliminaries}\label{sec:prelims}

In this section, we introduce some basic notation and definitions. 
Let $G = (V , E)$ be an undirected graph, where $V= \{1, \dots, n\}$.
Let $\emptyset \subsetneq S \subsetneq V$ be a subset of vertices. Any such vertex set is said to be a \emph{cut-set} of the graph. The \emph{cut} formed by $S$ is the set \[\cut_G(S) = \{e \in E \mid e = (u, v), u \in S, v \notin S\}.\]
The cut-sets $S$ and $V \setminus S$ both represent $\cut_G(S)$; thus, we sometimes refer to this cut by $(S, V \setminus S)$. \\

\noindent
Let $w:E \rightarrow \mathbb{Z}^+$ be an edge weight assignment. For a subset $E' \subseteq E$, we use $w\restriction_{E'}$ to denote the restriction of the weight function $w$ to the set $E'$. The \emph{weight} of a cut $(S,V\setminus S)$, denoted $w(\cut_G(S))$, is equal to the sum of the weights of the edges in $\cut_G(S)$. A cut $(S,V\setminus S)$ is a \emph{global minimum cut} of $G$ if its weight equals $\min_{\emptyset \subsetneq S \subsetneq V} w(\cut_G(S))$.  

\begin{fact}
    The function $w(\cut_G(\cdot)): 2^V \to \mathbb{Z}^+$ defined above is a submodular function. That is, for sets $S, T \subseteq V$, it holds that $w(\cut_G(S \cap T)) + w(\cut_G(S \cup T)) \leq w(\cut_G(S)) + w(\cut_G(T))$.
\end{fact}

\noindent
Let $s, t \in V$ be special source and sink vertices. A cut $(S, V\setminus S)$ is an $s$-$t$ cut if $s \in S$ and $t \notin S$. 
An $s$-$t$ cut $(S,V\setminus S)$ of $G$ is a \emph{minimum $s$-$t$ cut} if its weight equals $\min_{\emptyset \subsetneq S \subsetneq V, s \in S, t \notin S} w(\cut_G(S))$.\\

\noindent
We now introduce the concept of stitched weights, which allow us to perturb weights while maintaining integrality. This has been used in many prior works (eg.~\cite{FennerGurjarThierauf16, AnariVazirani19, GoldwasserGrossman17, SvenssonTarnawski17}).

\begin{definition}[Stitched Weights]\label{def:stitching}
Let $G = (V, E)$ be a graph and $w, w':E \rightarrow \mathbb{Z}^+$ be two edge weight assignments. Let $|w'|$ be any value such that $|w'| \geq \sum_{e \in E} w'(e)$. We define the stitching $[w, w']:E \rightarrow \mathbb{Z}^+$ to be an edge weight assignment defined as follows:
\[[w, w'](e) = (|w'| + 1)w(e) + w'(e)\]
\end{definition}

\noindent
For the sake of notational simplicity, while referring to the stitched weight of a subset $A \subseteq E$, we may denote $[w, w'](A)$ by a tuple $(w(A), w'(A))$. 
This notation is useful due to the observation that the ordering of stitched weights behaves exactly like the lexicographic ordering of tuples $(w, w')$.
\begin{observation}
\label{observation:stitching-induces-a-lexicographic-order}
Let $A, B \subseteq E$, then:
\begin{enumerate}
    \item $[w, w'](A) \leq [w, w'](B) \iff w(A) < w(B) \text{ or } \{w(A) = w(B) \text{ and } w'(A) \leq w'(B)\}$, and
    \item $[w, w'](A) = [w, w'](B) \iff w(A) = w(B) \text{ and }w'(A) = w'(B)$.
\end{enumerate}
\end{observation}

\noindent
~\Cref{observation:stitching-induces-a-lexicographic-order} immediately implies the following fact.
\begin{fact}
\label{fact: stitching-preserves-minimality-of-cuts}
Let $G = (V, E)$ be a graph and $w, w':E \rightarrow \mathbb{Z}^+$ be edge weight assignments. Any ($s$-$t$) minimum cut of $G$ under $[w, w']$ is also a ($s$-$t$) minimum cut of $G$ under $w$. 
\end{fact}

\begin{definition}[$\es$]
Let $V$ be the vertex set of a graph and $s \in V$ be a source vertex. The edge set $\es$ is defined to be the set of edges corresponding to the star graph on the vertex set $V$ with the vertex $s$ as the center. That is, 
\(\es = \{(s, v) \ | \ v \in V, v \neq s\}.\)
\end{definition}

\begin{definition}[$w_{s}$]
Let $G = (V, E)$ be a graph and $s \in V$ be a source vertex. The edge weight assignment $\wst:E \rightarrow \mathbb{Z}^+$ is defined by $e \mapsto \Ind{e \in \es}$. That is, 
\[ \wst(e) =  \begin{cases} 
       1, & \text{if } e \in \es \\
       0, & \text{otherwise.} \\
        \end{cases}
\]
\end{definition}

\noindent
Note that $n \geq \sum_{e \in E} \wst(e)$. Thus, we may set $|\wst|$ to $n$ when stitching $\wst$ (see~\Cref{def:stitching}).\\ 

\section{\texorpdfstring{Warm-up: Pseudodeterministic Minimum $s$-$t$ Cut}{Warm-up: Pseudodeterministic Minimum s-t Cut}}\label{sec:s-t-mincut}

In this section, we present our pseudodeterministic algorithm for the minimum $s$-$t$ cut problem and prove the following theorem. Specifically, we provide a generic transformation from randomized algorithms to pseudodeterministic algorithms.

\begin{theorem}
\label{theorem:pseudodeterministic-st-cut-algorithm-is-correct}
Let $\mathcal{A}$ be a randomized algorithm that:
\begin{enumerate}
    \item Takes as input a simple graph $G = (V, E)$ along with positive integer edge weights $w:E \rightarrow \mathbb{Z}^+$, and
    \item Outputs an $s$-$t$ cut $(S, T)$ of $G$ such that
    \[\Pr[(S,T) \text{ is a minimum {$s$-$t$} cut of }G \text{ under edge weights } w] \geq \rho\]
\end{enumerate}

\noindent
Then, the algorithm $\psdst$ is a pseudodeterministic algorithm which:
\begin{enumerate}
    \item Takes as input a simple graph $G = (V, E)$ along with positive integer edge weights $w:E \rightarrow \mathbb{Z}^+$, and
    \item There exists a fixed minimum $s$-$t$ cut $(S^*, T^*)$ of $G$ satisfying $\Pr[\psdst(G,w) = (S^*, T^*)] \geq \rho$.
\end{enumerate}
\end{theorem}

\begin{algorithm}[H]
\begin{algorithmic}[1]
\caption{\textsc{Pseudodeterministic Algorithm $\psdst$ for Minimum $s$-$t$ Cut}}
    \Require $G = (V, E)$ be the input graph with edge weights $w:E \rightarrow \mathbb{Z}^+$; oracle access to randomised algorithm $\mathcal{A}$ for weighted minimum $s$-$t$ cut
    \Ensure Minimum $s$-$t$ cut $(S,T)$

    \State Extend $w$ to the domain $E \cup \es$ by $w(e) = 0$ for all $e \in \es$.
    \State Run $\mathcal{A}$ on $G' = (V, E \cup \es)$ with edge weights $[w, \wst]$. Let the cut returned be $(S, T)$. 
    \State Output $(S, T)$.
    \end{algorithmic}
\end{algorithm}

\noindent
The following claim simply observes that adding weight $0$ edges to a graph does not change the value of any cuts.

\begin{claim}
\label{claim: adding-0-weight-edges-is-ok}
Let $G' = (V, E \cup \es)$ be a graph with edge weights $w:E \cup \es \rightarrow \mathbb{Z}^+$ such that $w(e) = 0 \ \forall e \notin E$. Then $(S, T)$ is a minimum $s$-$t$ cut of $G'$ if and only if $(S, T)$ is a minimum $s$-$t$ cut of $G = (V, E)$ with edge weights $w\restriction_E$. 
\end{claim}
\begin{proof}
Simply note that for any cut $(S, T)$ of $G'$, \[w(\cut_{G'}(S)) = \sum_{e \in \cut_{G'}(S) \cap E} w(e) + \sum_{e \in \cut_{G'}(S) \setminus E} w(e) = \sum_{e \in \cut_{G'}(S) \cap E} w(e) = \sum_{e \in \cut_G(S)} w\restriction_E(e)\]
\[ = w\restriction_{E}(\cut_G(S))\]

\noindent
Thus, because each cut in $G'$ with edge weights $w$ have exactly the same weight as the corresponding cut in $G$ with edge weights $w\restriction_E$, the minimum cuts of $G'$ with weights $w$ and $G$ with weights $w\restriction_E$ are exactly the same. 
\end{proof}

\begin{fact}
\label{fact: main-property-of-submodular-functions}
Let $G = (V, E)$ be a graph with edge weights $w:E \rightarrow \mathbb{Z}^+$ and a source vertex $s$ and sink vertex $t$. Let $S \subseteq V$ and $S' \subseteq V$ be cut-sets that both correspond to minimum $s$-$t$ cuts in $G$ under edge weights $w$, with $s \in S' \cap S$. Then $S \cup S'$ is also a cut-set corresponding to a minimum $s$-$t$ cut in $G$ under edge weights $w$.
\end{fact}
\begin{proof}
    It is a simple fact of submodular functions that $w(\cut_G(S \cap S')) + w(\cut_G(S \cup S')) \leq w(\cut_G(S)) + w(\cut_G(S'))$. Since $S$ and $S'$ are both minimisers of $w(\cut_G(S))$, both terms on the left must also be minimisers. Thus, $S \cup S'$ is also a cut-set corresponding to a minimum cut of $G$ under edge weights $w$.
\end{proof}

\begin{claim}
\label{unique-st-min-cut}
Let $G' = (V, E \cup \es)$ be a graph with edge weights $w:E \cup \es \rightarrow \mathbb{Z}^+$. Let $s \in V$ be a fixed source vertex. For any $t \in V \setminus \{s\}$, there exists a unique minimum $s$-$t$ cut in $G'$ with edge weights $[w, \wst]$.
\end{claim}
\begin{proof}
Fix some $t \in V \setminus \{s\}$.
Assume for the sake of contradiction that $G'$ under edge weights $[w, \wst]$ admits multiple minimum weight $s$-$t$ cuts. Let these cuts be $(S_1, T_1)$ and $(S_2, T_2)$. Assume that they have weight $(W, W')$. Note that $W' = |T_1| = |T_2|$. Consider now the cut $(S_1 \cup S_2, T_1 \cap T_2)$.

By \Cref{fact: stitching-preserves-minimality-of-cuts}, we know that $(S_1,T_1)$ and $(S_2,T_2)$ are both also minimum cuts of $G$ under edge weights $w$. This implies, by \Cref{fact: main-property-of-submodular-functions}, that $(S_1 \cup S_2, T_1 \cap T_2)$ is also a minimum cut of $G$ under edge weights $w$. Thus, \[w(\cut_G(S_1 \cup S_2)) = W.\] 

\noindent
Moreover, \[\wst(\cut_G(S_1 \cup S_2)) = |T_1 \cap T_2| < |T_1| = |T_2| = \wst(\cut_G(S_1)) = \wst(\cut_G(S_2)).\]

\noindent
Thus, we have
\[[w, \wst](\cut_G(S_1 \cup S_2)) < [w, \wst](\cut_G(S_1)) = [w, \wst](\cut_G(S_2))\]

\noindent
This contradicts the minimality of $(S_1, T_1)$ and $(S_2, T_2)$. This completes the proof.
\end{proof}

\begin{proof}[Proof of~\cref{theorem:pseudodeterministic-st-cut-algorithm-is-correct}]
Consider $G' = (V, E \cup \es)$ with edge weights $[w, \wst]$. By \Cref{unique-st-min-cut}, $G'$ with weights $[w, \wst]$ has a unique minimum $s$-$t$ cut. Let it be $(S^*, T^*)$. By \Cref{claim: adding-0-weight-edges-is-ok} and \Cref{fact: stitching-preserves-minimality-of-cuts}, $(S^*, T^*)$ is a minimum cut of $G$ with respect to weights $w$.\\

\noindent
Now, simply note that $\psdst(G, w) = \mathcal{A}(G', [w, \wst])$. With probability $\geq \rho$, the output $\mathcal{A}(G', [w, \wst])$ is a minimum cut of $G'$ under weights $[w, \wst]$. The only such minimum cut is $(S^*, T^*)$. Thus, with probability $\geq \rho$, the output $\psdst(G, w) = (S^*, T^*)$. This completes the proof. 
\end{proof}

\begin{remark}
As far as we are aware, the most efficient minimum $s$-$t$ cut algorithms in all models use $s$-$t$ max flow algorithms. Given an $s$-$t$ max flow, it is easy to obtain a canonical minimum $s$-$t$ cut by simply finding the cut-set $S$ of vertices reachable from $s$ in the residual graph formed by the maximum flow. Thus, we instead focus on global minimum cut, where much more efficient algorithms than max flow are known. 
\end{remark}

\section{Pseudodeterministic Global Minimum Cut}\label{sec:global-min-cut}

In this section, we design pseudodeterministic algorithms for the global minimum cut problem. Specifically, we prove the following. 

\begin{theorem}
\label{theorem:pseudodeterministic-global-cut-algorithm-is-correct}
Let $\mathcal{A}$ be a randomised algorithm which:
\begin{enumerate}
    \item Takes as input a simple graph $G = (V, E)$ along with positive integer edge weights $w:E \rightarrow \mathbb{Z}^+$, and
    \item Outputs a cut $(S, T)$ of $G$ such that
    \[\Pr[S \text{ is a minimum cut of }G \text{ under edge weights } w] \geq \rho\]
\end{enumerate}

\noindent
Then, the algorithm $\psdgb$ is a pseudodeterministic algorithm which:
\begin{enumerate}
    \item Takes as input a simple graph $G = (V, E)$ along with positive integer edge weights $w:E \rightarrow \mathbb{Z}^+$, and
    \item There exists a fixed minimum cut $(S^*, T^*)$ of $G$ such that \[\Pr[\psdgb(G,w) = (S^*,T^*)] \geq 1 - (1 - \rho) \cdot O(\log n).\]
    
\end{enumerate}
\end{theorem}

\noindent
The structure of global minimum cuts, for our purpose, is slightly more complicated than minimum $s$-$t$ cuts. The following example of a simple path illustrates why our approach for minimum $s$-$t$ cut does not work immediately.

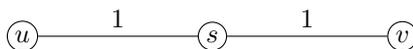
\begin{center}
    \begin{tikzpicture}[scale=1]
  \node (u) at (0,0) [circle, draw, inner sep=1.5pt] {$u$};
  \node (s) at (2.5,0) [circle, draw, inner sep=1.5pt] {$s$};
  \node (v) at (5,0) [circle, draw, inner sep=1.5pt] {$v$};

  \draw (u) -- node[above] {$1$} (s)
        (s) -- node[above] {$1$} (v);
\end{tikzpicture}
\end{center}

\noindent
Simply note that both $\{s, v\}$ and $\{s, u\}$ are global minimum cuts of the transformation $G' = (V, E \cup \es)$ under edge weights $[w, \wst]$. \\

\noindent
The following lemma sheds light on the structure of these counterexamples.

\begin{lemma}
\label{lemma:cuts-are-disjoint-and-equal-sized}
Let $G = (V, E)$ be a simple graph with edge weights $w:E \rightarrow \mathbb{Z}^+$. Let $s \in V$ be a special source vertex. Consider the transformed graph $G' = (V, E \cup \es)$ with edge weights $[w, \wst]$.
Let $\{(S_1, T_1),\dots, (S_k, T_k)\}$ be the global minimum cuts of $G'$ under edge weights $[w, \wst]$ such that $s \in S_i \ \forall i \in [k]$. Then:
\begin{enumerate}
    \item For all $i,j \in [k]$ such that $i \neq j$, we have $T_i \cap T_j = \emptyset$.
    \item For all $i,j \in [k]$, we have $|T_i| = |T_j|$. 
\end{enumerate}
\end{lemma}
\begin{proof}
\Cref{unique-st-min-cut} implies that the $T$s must be pairwise disjoint. Else, let $(S_i, T_i)$ and $(S_j, T_j)$ both be minimum cuts such that $s \in S_i \cap T_j$ and $t \in T_i \cap T_j$. These are both minimum $s$-$t$ cuts, which contradicts the uniqueness of the minimum $s$-$t$ cut. \\

\noindent
Now, consider any $(S_i, T_i)$ and $(S_j, T_j)$. Note that since these are both minimum cuts, in particular they have the same weight with respect to both $w$ and $\wst$ (by \Cref{observation:stitching-induces-a-lexicographic-order}). Thus, $\wst(\cut_{G'}(S_i)) = \wst(\cut_{G'}(S_j)) \implies |T_i| = |T_j|$. This completes the proof. 
\end{proof}

\noindent
In order to exploit this property, we now introduce a second type of edge set and weight assignment.


\begin{definition}[$\wi{i}$]
    Let $G = (V, E)$ be a graph and $s \in V$ be a source vertex. The weight assignment $\wi{i}:E \rightarrow \mathbb{Z}^+$ is defined as 
    \[ \wi{i}(e) =  \begin{cases} 
       1, & \text{if } e = (s, v) \text{ and }v<i \\
       0, & \text{otherwise.} \\
        \end{cases}
    \]
\end{definition}

\noindent
Note that $n \geq \sum_{e \in E} \wi{i}(e)$. Thus, we may set $|\wi{i}| = n$ when stitching $\wi{i}$ (see~\Cref{def:stitching}).\\ 

\noindent
We now need to work with three weight assignments instead of two, and so we introduce the concept of repeated stitching.

\begin{definition}[Repeated Stitching]
Let $G = (V, E)$ be a graph and $w_1, w_2, w_3:E \rightarrow \mathbb{Z}^+$ be edge weights. The weight assignment $[w_1, w_2, w_2]$ is defined as the repeated stitching $[[w_1, w_2], w_3]$. 
\end{definition}

\begin{claim}
\label{binary-search}
Let $G = (V, E)$ be a graph with edge weights $w$, and let $s$ be a source vertex. Let $G' = (V, E \cup \es)$, and let $\{(S_1, T_1), \dots, (S_k, T_k)\}$ be the set of global minimum cuts in $G'$ according to weights $[w, \wst]$, such that $s \in S_i \ \forall i \in [k]$. Furthermore, let $\mathcal{T} = \bigcup_{i \in [k]} T_i = \{t_1, t_2, \dots, t_{m}\}$. \\

\noindent
For any $x \in V$:
\begin{enumerate}
    \item If $x > t_{m}$, then every minimum cut $(S, T)$ of $G'$ w.r.t.\ weights $[w, \wst, \wi{x}]$ (with $s \in S$) satisfies $\wi{x}(\cut_{G'}(S)) = |T|$.
    \item If $x \leq t_{m-1}$, then either every minimum cut $(S, T)$ of $G'$ w.r.t.\ weights $[w, \wst, \wi{x}]$ (with $s \in S$) satisfies $\wi{x}(\cut_{G'}(S)) < |T| - 1$, or $G'$ under weights $[w, \wst, \wi{x}]$ has multiple minimum cuts (or both).
    \item If $t_{m-1} < x \leq t_m$, then the minimum cut $(S, T)$ of $G'$ w.r.t.\ weights $[w, \wst, \wi{x}]$ is unique and satisfies $\wi{x}(\cut_{G'}(S)) = |T| - 1$. Moreover, $(S, T)$ is exactly the cut $(S_i, T_i)$ such that $t_m \in T_i$.
\end{enumerate}
\end{claim}
\begin{proof}
By \Cref{lemma:cuts-are-disjoint-and-equal-sized} then, all $T_i$ are disjoint and equal sized. By \Cref{fact: stitching-preserves-minimality-of-cuts}, any minimum cut $(S, T)$ of $G'$ under weights $[w, \wst, \wi{x}]$, is a minimum cut of $G'$ under weights $[w, \wst$]. Thus, $(S, T) = (S_i, T_i)$ for some $i \in [k]$, and more importantly $\wi{x}(S) = \min_{i \in [k]} \wi{x}(S_i)$.\\

\noindent
For any cut $S_i$, we have $\wi{x}(\cut_{G'}(S_i)) = |T_i \cap \{v \in V \mid v < x\}| = |T_i \setminus \{v \in V \mid v \geq x\}| = |T| - |\{v \in T \mid v \geq x\}|$. \\

\noindent
Now, we handle the cases individually:
\begin{enumerate}
    \item If $x > t_m$, we have that $|\mathcal{T} \cap \{v \in V \mid v \geq x\}| = \emptyset$, which implies that $\forall i \in [k]$, $\wi{x}(\cut_{G'}(S_i)) = |T_i|$. Thus, $\wi{x}(\cut_{G'}(S)) = |T|$. 
    \item If $x \leq t_{m-1}$, we have $|\mathcal{T} \setminus \{v \in V \mid v \geq x\}| \geq 2$. 
    \begin{enumerate}
    \item If there exists $i \in [k]$ such that $T_i \cap \{v \in V \mid v \geq x\} \geq 2$, then $\wi{x}(\cut_{G'}(S)) \leq \wi{x}(\cut_{G'}(S_i)) \leq |T_i| - 2 < |T_i| - 1 = |T| - 1$. Thus, we end up in the first case.
    
    \item If instead we have that $\forall i \in [k]$, $|T_i \cap \{v \in V \mid v \geq x\}| \leq 1$, then since the $T_i$ are disjoint, there must exist $i, j \in [k]$ such that $|T_i \cap \{v \in V \mid v \geq x\}| = |T_j \cap \{v \in V \mid v \geq x\}| = 1$, and we thus have that both $S_i$ and $S_j$ are minimum cuts of $G'$ under weights $[w, \wst, \wi{x}]$. This is the second case, where the minimum cut is not unique.
    \end{enumerate}
    
    \item If $t_{m-1} < x \leq t_m$. Let $(S_i, T_i)$ be the cut such that $t_m \in T_i$. Simply note that for all $j \neq i$, $\wi{x}(\cut_{G'}(S_j)) = |T_j| = |T|$, whereas $\wi{x}(\cut_{G'}(S_i)) = |T_i| - 1 = |T| - 1$. Thus, we have that $(S, T) = (S_i, T_i)$ is the unique minimum cut of $G'$ under weights $[w, \wst, \wi{x}]$ and it satisfies $\wi{x}(\cut_{G'}(S)) = |T| - 1$.  
\end{enumerate}

\noindent
This completes the proof.
\end{proof}

\noindent
The last ingredient we need is an algorithm that checks whether the minimum cut of a graph is unique.
We describe it in~\cref{alg:uniqueness-test}.

\begin{algorithm}[H]
\begin{algorithmic}[1]
\Require Graph $G = (V, E)$ with edge weights $w:E \rightarrow \mathbb{Z}^+$; oracle access to  a randomised algorithm $\mathcal{A}$ for the global minimum cut problem on simple weighted graphs.

\Ensure A unique cut in $G = (V,E)$ w.r.t.\ $w$ if it exists and \textbf{False} otherwise
    \State Run $\mathcal{A}$ on $G$ with edge weights $w$. 
    \State Let the cut returned be $(S, T)$ such that $s \in S$.
    \State Let $G' = (V, E \cup \es)$. Extend $w$ to the domain $E \cup \es$ by $w(e) = 0$ for all $e \in \es \setminus E$.
    \State  Let $w_1:E \cup \es  \rightarrow \mathbb{Z}^+$ be the edge weight assignment defined by $w_1(e) = 1$ if $e = (s, v) \in \es \cap \cut_{G'}(S)$, and $w_1(e) = 0$, otherwise. 
    \State  Select an arbitrary vertex $t \in T$.
    \State  Let $\bar{G} = (V, E \cup E_{\text{star}}^t)$.
    \State  Let $w_2:E \cup E_{\text{star}}^t  \rightarrow \mathbb{Z}^+$ be the edge weight assignment defined by $w_2(e) = 1$ if $e = (t, v) \in E_{\text{star}}^t \cap \cut_{G'}(S)$, and $w_2(e) = 0$, otherwise. 
    \State Run $\mathcal{A}$ on $G'$ with edge weights $[w, w_1]$. Let the cut returned be $(S', T')$ such that $s \in S'$.
    \State Run $\mathcal{A}$ on $\bar{G}$ with edge weights $[w, w_2]$. Let the cut returned be $(\bar{S}, \bar{T})$ such that $s \in \bar{S}$. 
    \State \textbf{return} $(S,T)$ if $S = S' = \bar{S}$ and \textbf{False} otherwise.
\end{algorithmic}
\caption{$\uqtest$}
\label{alg:uniqueness-test}
\end{algorithm}

\begin{theorem}
\label{theorem:uq-test-is-correct}
Let $\mathcal{A}$ be a randomised algorithm which:
\begin{enumerate}
    \item Takes as input a simple graph $G = (V, E)$ along with positive integer edge weights $w:E \rightarrow \mathbb{Z}^+$, and
    \item Outputs an cut $(S, T)$ of $G$ such that
    \[\Pr[S \text{ is a minimum cut of }G \text{ under edge weights } w] \geq \rho\]
\end{enumerate}

\noindent
Then, $\uqtest$ is a randomised algorithm which:
\begin{enumerate}
    \item Takes as input a simple graph $G = (V, E)$ along with positive integer edge weights $w: E \rightarrow \mathbb{Z}^+$, and
    \item Decides whether the minimum cut of $G$ is unique with probability $\geq 1 - 3(1-\rho)$,
    \item Using only $O(1)$ calls to $\mathcal{A}$. 
\end{enumerate}
\end{theorem}
\begin{proof}
\noindent
First, note that $\uqtest$ makes exactly 3 calls to $\mathcal{A}$. Assume that both calls return a minimum cut. This happens with probability $\geq 1 - 3(1 - \rho)$.\\

\noindent
If the minimum cut $(S,T)$ of $G$ w.r.t.\ weights $w$ is unique, then the cuts $(S', T')$ and $(\bar{S}, \bar{T})$ must both be equal to $(S, T)$. This is true for $(S', T')$ because it is the minimum cut of $G' = (V, E \cup \es)$ under $[w, w_1]$, which by \Cref{fact: stitching-preserves-minimality-of-cuts} is a minimum cut of $G' = (V, E \cup \es)$ under $w$, which by \Cref{claim: adding-0-weight-edges-is-ok}, is a minimum cut of $G$ under $w$. The same proof works for $(\bar{S}, \bar{T})$. Thus, in this case, $\uqtest$ correctly decides that the minimum cut of $G$ under weights $w$ is unique.\\

\noindent
Assume that the minimum cut of $G$ under weights $w$ is not unique. Let $(S, T)$  be one of them, and let $(\hat{S}, \hat{T})$ be another. If $(S', T') \neq (S, T)$, then $\uqtest$ correctly decides that the minimum cut of $G$ under weights $w$ is not unique. Thus, assume that $(S', T') = (S, T)$. We will show that in this case $(\bar{S}, \bar{T}) \neq (S, T)$. First, we have $w(\cut_{G'}S) = w(\cut_{G'}\hat{S})$. Thus, if $(S', T') = (S, T)$, we must have $w_1(\cut_{G'}S) \leq w_1(\cut_{G'}\hat{S}) \iff |T| \leq |\hat{T} \cap T| \implies T \subsetneq \hat{T}$. Thus, note that we have $t \in \hat{T}$ and $|\hat{S}| < |S|$. This implies $w_2(\cut_{\bar{G}}S) = |S| > |\hat{S}| = w_2(\cut_{\bar{G}} \hat{S})$. Thus, $(\bar{S}, \bar{T}) \neq (S, T)$. This shows that $\uqtest$ correctly decides that the minimum cut of $G$ w.r.t.\ weights $w$ is not unique. 
\end{proof}

\noindent
Next, we describe our pseudodeterministic global minimum cut algorithm in~\cref{alg:global-min-cut}. It runs both an arbitrary global minimum cut algorithm $\mathcal{A}$ as well as $\uqtest$. We then complete the proof of~\cref{theorem:pseudodeterministic-global-cut-algorithm-is-correct}. 

\begin{algorithm}[H]
\begin{algorithmic}[1]

\Require Graph $G = (V, E)$ with edge weights $w:E \rightarrow \mathbb{Z}^+$; oracles access to $\mathcal{A}$, a randomised algorithm for the global minimum cut problem on simple weighted graphs

\Ensure Global Minimum Cut $(S,T)$
    \State Fix an arbitrary vertex $s \in V$ to be the source vertex. 
    \State Extend $w$ to the domain $E \cup \es$ by $w(e) = 0$ for all $e \in \es \setminus E$.
    \State Define $G' = (V, E \cup \es)$.
    \State If $G'$ has a unique global minimum cut $(S, T)$ under weights $[w, \wst]$, output $(S, T)$ and terminate. \label{step:unique-cut-return}
    \State $l \gets 1, u \gets n$
    \While {$l \leq u$} 
\Comment{Do binary search over the range $[l = 1, u = n]$.}
   
       { \item Let $x = \lfloor (l+u)/2 \rfloor$ and $(S, T) = \mathcal{A}(G', [w, \wst, \wi{x}])$
        \item If $(S, T)$ is the unique global minimum cut in $G'$ with weights $[w, \wst, \wi{x}]$, and $\wi{x}(\cut_{G'}(S)) = |T| - 1$, output $(S,T)$ and terminate. \label{step:unique-cut-return-2}
        \item If $\wi{x}(\cut_{G'}(S)) = |T|$, let $u = x-1$ and continue the binary search.
        \item Else, let $l = x+1$ and continue the binary search.}
    \EndWhile

\end{algorithmic}
\caption{\textsc{Pseudodeterministic Algorithm $\psdgb$ for Global Minimum Cut}}
\label{alg:global-min-cut}
\end{algorithm}

\begin{proof}[Proof of~\Cref{theorem:pseudodeterministic-global-cut-algorithm-is-correct}]
\noindent
Assume that every time $\psdgb$ uses a randomised algorithm as a subroutine, it returns the correct answer. This happens with probability $\geq \rho$ for each use of $\mathcal{A}$, and with probability $\geq 1 - (1-\rho)\cdot 3$ for the uniqueness test. Since we use the uniqueness test and $\mathcal{A}$ at most $O(\log n)$ many times, by a union bound, all subroutines work correctly with probability $\geq 1 - (1 - \rho) \cdot O(\log n)$.\\

\noindent
Now, let $(S_1, T_1), \dots, (S_k, T_k)$ be the minimum cuts of $G' = (V, E \cup \es)$ under weights $[w, \wst]$ (such that $\forall i \in [k], \ s \in S_i )$, and let $\mathcal{T} = \bigcup_{i \in [k]} T_i = \{t_1, \dots, t_m\}$. Let $i \in [k]$ be such that $t_m \in T_i$. \Cref{lemma:cuts-are-disjoint-and-equal-sized} shows that the $T_j$ are all disjoint, and thus this choice of $i$ is unique. We will show that $\psdgb$ returns $(S^*, T^*) = (S_i, T_i)$.\\

\noindent
First, if $k=1$, then~\cref{step:unique-cut-return} of~\cref{alg:global-min-cut}  necessarily returns $(S^*, T^*) = (S_i, T_i)$. Thus, assume that $k>1$. In this case, we perform the binary search. Now, simply note that the conditions of \Cref{binary-search} ensure that our binary search procedure finds $x \in (t_{m-1}, t_m]$, at which point $G'$ w.r.t.\ weights $[w, \wst, \wi{x}]$ has a unique minimum cut, which is $(S_i, T_i)$, and $\psdgb$ algorithm outputs this in~\cref{step:unique-cut-return-2}. \\

\noindent
Thus, if every randomised subroutine works correctly, $\psdgb$ outputs $(S^*, T^*)$. As discussed in the first paragraph of the proof, this happens with probability $\geq 1 - (1 - \rho) \cdot O(\log n)$. This completes the proof. 
\end{proof}

\section{Streaming, Cut Query, and Parallel Models} \label{sec:implementation}
In this section, we will briefly argue that \Cref{alg:global-min-cut} can be implemented efficiently across sequential, streaming, cut query, and parallel models. 

\subsection{Sequential}
Say $\mathcal{A}$ runs in time $T(n, m)$ on a graph $G$ with $n$ vertices and $m$ edges. Clearly, \Cref{alg:uniqueness-test} can be implemented such that it runs in $O(m + n + T(n, m + n))$ time. Then, \Cref{alg:global-min-cut} can be implemented to run in $O((m + n + T(n, m + n)\log n \log\log n))$ time. Combined with the randomized $O(m\log^2n)$ time algorithm of \cite{GawrychowskiMozesWeimann24}, this gives us a pseudodeterministic $O(m\log^3n\log\log n)$ time algorithm for global minimum cut, thereby proving \Cref{thm:sequential}. This additional $\log \log n$ factor is in order to boost the success probability of $\mathcal{A}$ from $2/3$ to $1 - 1/\Omega(\log n)$. We incur a similar cost in the other models, but it is hidden under the tilde notation. 

\subsection{Parallel}
Both \Cref{alg:global-min-cut} and \Cref{alg:uniqueness-test} involve only the following non-trivial steps:
\begin{enumerate}
    \item Adding the edges $\es$ of a star graph to the graph $G$. This can be done with $\tilde{O}(n+m)$ work and $\tilde{O}(1)$ depth by simply adding $(s, v)$ for all $v \in V$ in parallel. It is easy to avoid adding a duplicate edge by maintaining a lookup table indexed by $v \in V$, indicating whether $(s, v) \in E$. 
    \item Stitching a weight assignment $w'$ with a weight assignment $w$. Our algorithm only does this when $n$ is an upper bound on $\sum_{e \in E} w'(e)$. Thus, we simply need to replace the weight $w$ of each edge with $(n+1) \cdot w(e) + w'(e)$. By doing this in parallel for each $e \in E$, this takes $\tilde{O}(n+m)$ work and $\tilde{O}(1)$ depth.
    \item Testing whether two cut-sets $S$ and $S'$ are equal. There are many easy ways to do this with $\tilde{O}(n+m)$ work and $\tilde{O}(1)$ depth. 
    \item Computing the weight of a cut with respect to an edge weight assignment $w$. By simply iterating over all edges in parallel, and then adding up their contribution to the cut in a tree-like fashion, this can be done with $\tilde{O}(n+m)$ work and $\tilde{O}(1)$ depth. 
\end{enumerate}

\noindent
Say $\mathcal{A}$ takes work $W(n, m)$ and depth $D(n, m)$ on a graph with $n$ vertices and $m$ edges. With the aforementioned observations, it's easy to see that \Cref{alg:uniqueness-test} and \Cref{alg:global-min-cut} can be implemented with work $\tilde{O}(n + m + W(n, m + n))$ and time $\tilde{O}(D(n, m + n))$. The $\tilde{O}$ absorbs the $\log n$ factor from the binary search in \Cref{alg:global-min-cut}. Combined with the $\tilde{O}(n + m)$ work and $\tilde{O}(1)$ depth algorithm of~\cite{AndersonBlelloch23}, we obtain a pseudodeterministic $\tilde{O}(m)$ work and $\tilde{O}(1)$ depth parallel algorithm, thereby proving \Cref{thm:parallel}. 

\subsection{Streaming}
Both \Cref{alg:global-min-cut} and \Cref{alg:uniqueness-test} involve only the following non-trivial steps:
\begin{enumerate}
    \item Adding the edges $\es$ of a star graph to the graph $G$. This can be done by simply suffixing the edges $\es \setminus E$ to the input stream with weights $0$. In order to do this, one needs to track the edges $E \cap \es$ in the input stream. This is easy to do using $O(n \log n)$ space since $|\es| = n$ -- one can simply maintain a list of all edges of $\es$ and whether they have already appeared in the stream.
    \item Stitching a weight assignment $w'$ with a weight assignment $w$. Our algorithm only does this when $n$ is an upper bound on $\sum_{e \in E} w'(e)$. Moreover, our algorithm can always compute $w'(e)$ for any edge $e$ given only the endpoints of $e$. Thus, every time we encounter an edge $e$ in the stream with edge $w(e)$, we can simply change its weight to $w(e) \cdot (n+1) + w'(e)$. This part implicitly requires space to store the information to compute $w'$, which is just $s$, if $w' = \wst$, or $x$ and $s$, if $w' = \wi{x}$, or the cut-set $S$ along with vertices $s$ and $t$ if $w'$ is $w_1$ or $w_2$. Storing the cut-set $S$ is the most expensive and only requires $O(n \log n)$ space.
    \item Testing whether two cut-sets $S$ and $S'$ are equal. Since we have no time restriction in the streaming model, once we have the sets $S$ and $S'$ stored using $O(n \log n)$ space, this is completely trivial to do using no additional space.  
    \item Computing the weight of a cut $(S, T)$ with respect to a weight assignment $w'$. In our algorithm, we only need to do this when $w' = \wi{x}$, and the algorithm already knows $x$. Thus, this is easy to do: we maintain $W$, initialised to $0$. Every time an edge $e$ crossing the cut $(S, T)$ is added to the graph in the stream, we modify $W = W + w'(e)$, and every time an edge $e$ is deleted from the graph, we modify $W = W - w'(e)$. The only space required here is to store $x$ and $W$. This requires $O(\log n)$ space. 
\end{enumerate}

\noindent
Let $\mathcal{A}$ be a streaming algorithm that uses $S(n, m)$ space and makes $P(n, m)$ passes over the stream. With the aforementioned observations, \Cref{alg:uniqueness-test} and \Cref{alg:global-min-cut} can be implemented with space $\tilde{O}((S(n, m) + n)$ and $O(P(n, m) \cdot \log(n))$ passes over the stream. Combined with the $\tilde{O}(n)$ space and $O(\log n)$ passes algorithm of~\cite{MukhopadhyayNanongkai20}, we obtain a pseudodeterministic $\tilde{O}(n)$ space and $O(\log^2n)$ passes streaming algorithm, thereby proving \Cref{thm:streaming}.

\subsection{Cut Query}
For cut query algorithms, the argument is largely very simple, since the only restriction in the model is on queries to the graph and not on any computational aspect. The only non-trivial implementation question is about how to stitch weights $[w, w']$ in the cut query model. Formally, given a cut oracle $\mathcal{O}$, where $\mathcal{O}(S) = w(\cut_{G}(S))$, we want to design an oracle $\mathcal{O}'$ where $\mathcal{O}'(S) = [w, w'](\cut_{G}(S))$. This is easy to do because of our choices of stitched weights and graphs. In both \Cref{alg:uniqueness-test} and \Cref{alg:global-min-cut}, $w'(\cut_{G}(S))$, for any $\emptyset \subsetneq S \subsetneq V$, can be computed given only $S$, and the parameters of the weight function $w'$. Assume $s \in S$:
\begin{enumerate}
    \item $\wst(\cut_{G}(S))$ is equal to $n - |S|$. 
    \item $\wi{x}(\cut_{G}(S))$ is equal to $|\{v \in V \mid v < x\} \cap (V \setminus S)|$.
    \item Let $\hat{S}$ be the cut-set with respect to which $w_1$ is defined, such that $s \in \hat{S}$. Then $w_1(\hat{S}) = |(V \setminus S) \cap (V \setminus \hat{S})|$. The same argument holds for $w_2$ and $t$. 
\end{enumerate}

\noindent
Thus, the oracle $\mathcal{O}'$ can simply return as its answer $\mathcal{O}'(S) = \mathcal{O}(S) \cdot (n+1) + w'(\cut_{G}(S))$\\

\noindent
Let $\mathcal{A}$ be a cut-query algorithm that makes $C(n, m)$ many cut queries. With the aforementioned observation, it's easy to see that \Cref{alg:global-min-cut} and \Cref{alg:uniqueness-test} can be implemented such that they require $O(C(n, m) \cdot \log n)$ many queries. Thus, combined with the algorithm of~\cite{MukhopadhyayNanongkai20}, we obtain an algorithm which requires $\tilde{O}(n)$ many cut queries, thereby proving \Cref{thm:cut-query}.

\DeclareUrlCommand{\Doi}{\urlstyle{sf}}
\renewcommand{\path}[1]{\small\Doi{#1}}
\renewcommand{\url}[1]{\href{#1}{\small\Doi{#1}}}
\newpage
\bibliographystyle{plainurl}
\bibliography{bibliography}
\newpage 
\appendix
\end{document}